\documentclass[conference]{IEEEtran}
\usepackage[english]{babel}
\usepackage[utf8x]{inputenc}
\usepackage{amssymb,amsfonts,amsthm}
\usepackage{bbold}
\usepackage{dsfont}
\usepackage{graphicx}
\usepackage{blkarray}
\usepackage[colorinlistoftodos]{todonotes}
\usepackage{amsmath}
\usepackage{url}
\newtheorem{mydef}{Definition}
\newtheorem{proposition}{Proposition}
\newtheorem{Lemma}{Lemma}

\ifCLASSINFOpdf
\else
\fi
\usepackage{amsthm}
\usepackage{cite}
\usepackage{balance}
\usepackage{latexsym}
\usepackage{graphicx}
\usepackage{caption}
\usepackage{soul,color}
\usepackage{amsmath}
\usepackage{algorithm}
\usepackage{subfigure}
\usepackage{authblk}
\usepackage[noend]{algpseudocode}
\newtheorem{thm}{Theorem}

\IEEEoverridecommandlockouts

\usepackage{textcomp}
\usepackage{xcolor}

\begin{document}{}

\title{Whittle Index Policy for Multichannel Scheduling in Queueing Systems}


\author[1]{\small Saad Kriouile}\author[2]{\small Maialen Larranaga}\author[1]{\small Mohamad Assaad}
\affil[1]{TCL Chair on 5G, Laboratoire des Signaux et Syst\`emes CentraleSup\'elec,  91192 Gif sur Yvette, France}\affil[2]{ASML, P.O. Box 324, 5500 AH Veldhoven, The Netherlands}

\maketitle
\newcommand{\HRule}{\rule{\linewidth}{0.5mm}} 

 


 

~







\begin{abstract}


In this paper, we consider a queueing system with multiple channels (or servers) and multiple classes of users. We aim at allocating the available channels among the users in such a way to minimize the expected total average queue length of the system. This known scheduling problem falls in the framework of Restless Bandit Problems (RBP) for which an optimal solution is known to be out of reach for the general case. The contributions of this paper are as follows. We rely on the Lagrangian relaxation method to characterize the Whittle index values and to develop an index-based heuristic for the original scheduling problem. The main difficulty lies in the fact that, for some queue states, deriving the Whittle's index requires introducing a new approach which consists in introducing a new expected discounted cost function and deriving the Whittle's index values with respect to the discount parameter $\beta$. We then deduce the Whittle's indices for the original problem (i.e. with total average queue length minimization) by taking the limit 
$\beta \rightarrow 1$. The numerical results provided in this paper show that  this policy performs very well and is very close to the optimal solution for high number of users.

\end{abstract}

\section{Introduction}

The problem of scheduling and resource allocation have been widely recognized as a way to improve the network performance and meet the service requirements in networks. Many resource allocation problems have been studied in the past in wired and wireless networks. In this paper, we are interested in the problem of scheduling in queueing systems where a set of users or queues share a set of servers. At each time slot, the servers are allocated to the users in such a way to minimize the total expected length of the users' queues.  Although this problem is well known in the literature, one can show that it is a Restless Bandit Problem (RBP), which is very hard to solve as we will see in the sequel. In fact, this problem has been well studied in the past from a stability perspective  \cite{destounis},  \cite{tassiulas},  \cite{georgiadis}. It has been shown that max weight policy is throughput optimal and many variants have been proposed to deal with different settings and conditions  \cite{tassiulas},  \cite{georgiadis}. The main weakness of max weight policy is that it may result in a high (but finite) average delay. In order to improve the average delay in the system, we are interested in minimizing the total average length of the queues. This problem is a hard problem and can be cast as a Restless Bandit Problem (RBP), a particular model of Markov Decision Processes (MDP).  RBPs are PSPACE-Hard see Papadimitriou et al. \cite{papadimitriou}, which means that their optimal solution is out of reach. One has therefore to develop a sub-optimal but well performing policy.   In this paper, We propose a Whittle index policy to deal with the aforementioned problem. The development of such policy is not straightforward and requires some analysis to prove that such policy exists and to make the corresponding derivation. First, we introduce a new discount factor in the reward function (denoted by $\beta$), analyze the Lagrangian relaxation of the  resulting discounted reward problem, deriving the Whittle's index for this new relaxed problem as a function of $\beta$ and then taking the limit when $\beta \rightarrow 1$ to find the Whittle's index of the original problem. The interest of finding the Whittle's index expressions for our problem is that we can use the known Whittle's index policy (WIP) to allocate the resources to users. In this paper, we will show that for our model an explicit expression of Whittle's index can be found. WIP has been proposed as a suboptimal policy for many problems in the literature, see for instance \cite{liu,ansell}. It has also been shown to perform near optimally in many scenarios and in the particular case of multiclass M/M/1 queues, WIP which simplifies to the $c\mu$-rule is optimal, see  \cite{buyukkoc} and  \cite{larranaga15}. Therefore, WIP (when it is possible to obtain it) is a well performing policy. This motivates the development of the Whittle's index expressions  in this paper. 

\subsection{Related Work}
There are lot of works which study the problem of resource allocation in wireless networks. For instance, in \cite{deghel}  \cite{destounis} \cite{tassiulas}  \cite{georgiadis}, the authors give a throughput optimal policy for single channel, multi-channel and multi-user MIMO contexts using max weight rule, which is known to not be delay optimal. To overcome this matter, many works have been developed in the past to minimize the average delay of the traffic of the users (e.g. see \cite{cui12} and the references therein). Most of them describe the minimization problem as Markov Decision Process (MDP) and develop resource allocation policies using Bellman equation such as Value iteration algorithm. However, as we have already mentioned in abstract, MDP frameworks and Bellman equation are hard to solve them.   \cite{wang} \cite{cui10}   try to minimize the average delay of the users' queues using stochastic learning algorithm. Indeed, the stochastic learning algorithm consumes lot of time and users memories. Besides, it requires high  computational complexity. 

On the other hand,  for some  MDP problems, the optimal policy turns out to be reachable and has a form of index policy. For instance, in multi-class single-server queue with linear holding costs, the optimal index policy is the $c\mu$ rule that schedules the user with the highest $c\mu$, see \cite{buyukkoc}. Another classical result that can be seen as an index policy is the optimality of Shortest-Remaining-Processing-Time (SRPT), where the index of each customer is given by its remaining service time \cite{schrage}. Both examples fit the general context of Multi-Armed Bandit Problems (MABP), which is a particular case of MDP: at each decision time, we select only one bandit and its state evolves stochastically while the other bandits states stay unchanged. The aim of scheduler is to maximize the total average reward or to minimize the total average cost. 
Whittle introduces the called Restless Bandit (RBP) where the scheduler selects a fixed number of bandits, and all bandits (either active or not)  might evolve. He defines the Whittle index and the Whittle index policy, and prove that it is asymptotically optimal under some conditions.
However, in order to calculate Whittle’s index, there are two main difficulties: first, we need to
establish indexability, and second, the calculation of the Whittle
index itself might be infeasible in some  cases.
Whittle index policy has been derived for birth-and-death multi-class multi servers queue  in \cite{larranaga16}. In \cite{van}, an optimal index policy called Generalized$c\mu$-rule
(Gcµ) is developed in the context of heavy-traffic regime with convex delay cost.  Furthermore, in contrast to $c\mu$ rule policy, \cite{mandelbaum} establishes the optimality of Generalized$c\mu$-rule
(Gcµ) even with multiple servers. In \cite{ansell} the authors calculate Whittle’s index policy for a multi-class queue with general holding cost functions. 

The aforementioned cited works consider a time continuous model. \cite{kriouile} considers different MDP model with discrete time slotted system, and with finite buffer length. In this paper, we consider the same model, but the buffer length is infinite. The major difference between our work and \cite{kriouile} will be in the whittle index derivation. In fact, we will use a new discounted cost approach in order to derive whittle index. 
We will explain the difference with respect to \cite{kriouile} in more details in Section~\ref{sec:WI}. In this paper, we  provide an explicit characterization of Whittle's indices by introducing a discounted cost approach, and develop a Whittle index allocation policy for our original problem (average cost case) by adapting the Whittle's indices expressions when discounted parameter is near to one. We find that optimal solution can be seen as $c\mu$ rule for large queue state where $c$ is replaced by the weighted factor $a$ and $\mu$ is replaced by the  transmission rate $R$ which represents the number of packets that can be transmitted per time slot.


The remainder of the paper is organized as follows. In Section~\ref{sec:systemmodel}, we describe the system model and formulate the average cost minimization problem. In Section~\ref{sec:relProb&ThrePol}, we introduce the Lagrangian relaxation and show the optimality of threshold/monotone policies for the relaxed dual problem. In Section~\ref{sec:WI}, we characterize Whittle's indices explicitly for all queue states and explain Whittle's index policy. Numerical results are provided in Section~\ref{sec:numerics} and Section VI concludes the paper. The proofs are provided in the appendix.

\section{System Model}\label{sec:systemmodel}
\subsection{System model description}\label{sec:modelDes}
We consider a time-slotted system with one central scheduler, $N$ queues and $M$ uncorrelated "channels" or "servers" ($N > M$). The words channels and servers will be used interchangeably. We consider a discrete slotted time system, where at each time slot, the scheduler chooses $M$ users among $N$ and allocates to each one exactly one channel.
Let class-$k$ denotes the class of users for which user can transmit at most $R_k$ packets per time slot, called maximum transmission rate, if a channel/server is assigned to the user. In other words, the server rate is not fixed for all queues and depends on the class of users. We consider that the number of different classes is $K$.  and that for all $k$ $R_k \geq 2$, (we will give later brief justification of this assumption).  
Let $\gamma_k$  denotes the proportion of users in class $k$ with respect to the total number $N$ of users. We will use the terms users and queues interchangeably in this paper. We further denote by $A_{i}^k(t)\in\{0,\ldots,R_{k}-1\}$ the number of packets that arrive to class-$k$ queue $i$ at time slot $t$. From above, it is clear that for all $k$ $R_{k}-1 \geq 1$ or equivalently $R_k \geq 2$. Moreover, we assume that the packets arrival follows a uniform distribution. Therefore the probability that $k$ packets arrive at time slot $t$ is $\mathds{1}_{\{k \in \{0,\ldots,R_{k}-1\}\}} \rho_k$, with $\rho_k=1/R_k$.     
We denote by $s^{k}_{i}(t)$ the transmission decision as follows: $s^{k}_{i}(t)=1$ when the user $i$ in class $k$ is scheduled, and $s^{k}_{i}(t)=0$ otherwise.
Let $q^{k}_{i}(t)$ denotes the number of packets in queue $i$ belonging to class-$k$. We consider that all users have infinite queue length, \textit{which is the main difference with respect to the work in \cite{kriouile} in which a finite queue length is assumed}. Then we have: 
\begin{equation}\label{eq: qeue_evolution}
q^{k}_{i}(t+1)=(q^{k}_{i}(t)-R_ks^{k}_{i}(t))^{+}+A^k_i(t).
\end{equation}

\subsection{Problem formulation}\label{sec:probForm} 
We denote  by $\Phi$ the broad class of scheduling policies that make a scheduling decision based on the history of observed queue states and scheduling actions. Therefore,  the scheduling 
problem consists on finding a policy in $\Phi$ that minimizes
the infinite horizon expected average queue length,
subject to the constraint on the number of users selected
in each time slot, i.e. the number of scheduled users must not exceed the number of available channels. According to little law, minimizing the average queue length will reduce the average delay experienced by the users.  
Denoting $\frac{M}{N}$ by $\alpha$, the weighted factor for each class $k$ by $a_k$, and given the initial state $q(0)=(q^1_1(0),\ldots,q^1_{N\gamma_1}(0),...,q_{1}^K(0),\ldots,q^K_{N\gamma_K}(0))$, then the problem is formulated as:
\begin{equation}\label{eq:constraint}
\begin{aligned}
& \underset{\phi \in \Phi}{\text{min}}&& \limsup\limits_{T\rightarrow\infty}\frac{1}{T} E[ \sum_{t=0}^{T-1}  \sum_{k=1}^{K} \sum_{i=1}^{{\gamma_k}N} a_kq_i^k(t) \mid q(0), \phi], \\
& \text{s.t.} && \sum_{k=1}^{K} \sum_{i=1}^{{\gamma_k}N} s_i^k(t)\leq {\alpha}N, \forall t.\\
\end{aligned}
\end{equation}

\section{Relaxed Problem and Threshold Policy}\label{sec:relProb&ThrePol}

The problem described in Section~\ref{sec:probForm} is a Restless bandit problem (RBP)  since it consists in scheduling at each time $M$ users (or resources) among  the N users and that at each time the state of each  queue evolves even if the queue is not scheduled. See Whittle \cite{whittle}. Since RBPs are PSPACE-Hard. See Papadimitriou et al. \cite{papadimitriou}, therefore we need to develop an new approximation in order to derive well performing policies.
For that, we will first analyze a relaxed version of the original problem and then use the structure of its optimal policy to develop a Whittle index policy for the original problem. The relaxation considered here is the Lagrangian relaxation approach. This latter consists of relaxing the constraint of available servers. In other words, we consider that the constraint in Equation~\eqref{eq:constraint}, has to be satisfied on average and not in every time slot. That means:
\begin{equation}
\limsup\limits_{T\rightarrow\infty}\frac{1}{T} E[\sum_{k=1}^{K} \sum_{i=1}^{{\gamma_k}N} s_i^k(t)] \leq {\alpha}N.
\end{equation}
Denoting $W$ by the Lagrangian multiplier for the constrained problem, then the Lagrange function equals to:
\begin{align*}
f(W,\phi)=& \limsup\limits_{T\rightarrow\infty}\frac{1}{T} E[ \sum_{t=0}^{T-1}  \sum_{k=1}^{K} \sum_{i=1}^{{\gamma_k}N} (a_kq_i^k(t)+Ws_i^k(t)) \mid \phi,q(0)]\\
&-W{\alpha}N
\end{align*}
Where $W$ can be seen as a subsidy for not transmitting, or the price to decide an active action.
Therefore, the dual problem for a given $W$ is
\begin{equation}
\underset{\phi \in \Phi}{\text{min}} \ f(W,\phi).
\end{equation}
\subsection{Problem Decomposition}
The relaxed problem allows  to decompose the $N$-dimensional problem into much simpler 1-dimensional subproblems. For that, we fix the Lagrangian parameter $W$ and discard from the dual problem formulation the sum which does not depend on $\phi$ (since the problem considered is an optimization problem over a set of policies $\Phi$). Hence, the dual problem will be equivalent to:  
\begin{equation}\label{eq:relaxed}
\underset{\phi \in \Phi}{\text{min}}\limsup\limits_{T\rightarrow\infty}\frac{1}{T} E[ \sum_{t=0}^{T-1}  \sum_{k=1}^{K} \sum_{i=1}^{{\gamma_k}N} (a_kq_i^k(t)+Ws_i^k(t)) \mid \phi,q(0)].
\end{equation}   
In fact, the solution of this problem is the stationary policy that resolves the well known Bellman equation, e.g. see Ross~\cite{ross}. Namely,
\begin{equation} \label{eq:Bellman}
V(q)+ \theta= \underset{s}{\text{min}} \{ \sum_{k=1}^{K} \sum_{i=1}^{{\gamma_k}N} C_k(q_i^k,s_i^k)+ \sum_{q'} Pr(q'|q,s)V(q')\}
\end{equation}
Where $V(\cdot)$ represents the value function, $\theta$ is the optimal average cost and $C_k(q_i^k,s_i^k)$ is the holding cost $a_kq_i^k+Ws^k_i$ in class-$k$. The optimal decision for each state $q$ can be obtained by minimizing the right hand side of Equation~\eqref{eq:Bellman}.
One can show that, for a given $W$, this relaxed problem can be decomposed into $N$ independent subproblems. We skip the proof here for brevity and refer the reader to \cite{kriouile},as the model therein is similar to our model here except that the queues there have a limited capacity.

\subsection{Threshold policy}\label{sec:threshold policy}
In this section, we show that the solution for each individual problem (for each user $i$) is a threshold policy.
We give first some useful definitions. 
\begin{mydef} 
For given class-$k$, a threshold policy is a policy $\phi \in \Phi$ for which there exists an $n \geq -1$ such that when the queue of user $i$ is in state $ q_i^k \leq n$, the prescribed action is $s^- \in \{0,1\}$. And when the queue $ q_i^k > n$, the prescribed action is $s^+ \in \{0,1\}$ and $s^- \neq s^+$.\\
Since there are only two possible actions, a policy is of the form threshold policy if and only if it is monotone in $q_i^k$.
\end{mydef}

\begin{mydef}
We say that function $f$ is R-convex in $X=[0,+\infty]$, if for any $x$ and $y$ in $X$ such that $x < y$, we have: $$f(y+R)-f(x+R) \geq f(y)-f(x)$$ 
\end{mydef}

\begin{mydef}
Let $g(x,y)$ be a real valued function defined on $X \times S$, with $S=\{0,1\}$, and $X=[0,+\infty]$. We say that $g$ is submodular if $g(x+1,1)-g(x+1,0) \leq g(x,1)-g(x,0)$ for all $x$ on $X$.
\end{mydef}

The solution of Bellman equation ~\eqref{eq:Bellman} $V(\cdot)$ can be obtained by an algorithm called Value iteration. This consists in updating $V_t(\cdot)$ by the following equation
\begin{equation} \label{value_iteration}
V_{t+1}(q_i^k)= \underset{s_i^k}{\text{min}} \{ C(q_i^k,s_i^k)+ \sum_{q^{'k}_i} Pr(q^{'k}_i|q_i^k,s_i^k)V_t(q^{'k}_i)\} - \theta_k
\end{equation}
After many iteration, $V_t(\cdot)$ will converge to the unique fixed point of the equation ~\eqref{eq:Bellman} called $V(\cdot)$.

\begin{mydef}
We define the operator $TO$ such that for each $(q_i^k,s_i^k) \in [0,+\infty]  \times \{0,1\} $ $$(TO(V))(q_i^k,s_i^k)\triangleq C(q_i^k,s_i^k)+ \sum_{q^{'k}_i} Pr(q^{'k}_i|q_i^k,s_i^k)V(q^{'k}_i)-\theta_k$$ 
\end{mydef}   

\begin{proposition}
For each class-$k$ and user $i$, the optimal solution that resolves the Bellman equation ~\eqref{eq:Bellman} is of type increasing threshold: there exists a state $n$ such that for each state $q_i^k \leq n$ the optimal decision is passive action, and for each state $q_i^k > n$ the optimal decision is active action.  
\end{proposition}
In order to prove this result we need to prove that $TO(V)$ is submodular.\\ 
Proof outline: Since our model is similar to the one considered in \cite{kriouile}, the proof is similar. We provide here a high level description of the proof:\\
1) One has to establish that for all $t$, $V_t(.)$ is increasing and R-convex (this can be done  by induction). From that one can conclude that $V(.)$ is also increasing and R-convex.\\
2) Demonstrate that if $V(.)$ is increasing and R-convex, then $TO(V)$ is submodular.\\
3) Conclude that the optimal solution is an increasing threshold in queue state, by exploiting the submodularity of the function $TO(V)$

\section{Whittle's index}\label{sec:WI}
In this section, we will introduce the notion of Whittle's index, which will be useful to develop a new heuristic for original problem. We will review the main result and approach obtained in \cite{kriouile} and explain its limitation and why such approach cannot be used to find the Whittle index values if the queues have unlimited capacity. 

At given state $n$, the Whittle's index is the Lagrange multiplier or subsidy for passivity for which the optimal decision at this state is indifferent (passive and active decision are both optimal). This definition requires that the property of indexability is satisfied. 
This property consists in establishing that as the subsidy for passivity, W, increases,
the collection of states in which the optimal action is passive increases.

Before providing a rigorous definition of indexability, we recall from the previous section that the optimal policy for the relaxed problem for given W is a threshold policy. We denote by $u^n_k$ the stationary distribution of the states under threshold policy $n$ in class-$k$. We now formalize the concepts of indexability and Whittle's index in the following definition.
\begin{mydef}
A class of queues is indexable if the set of states in which the passive action is the optimal action 
(denoted by $D(W)$) increases in $W$. That is, $W' < W \Rightarrow D(W') \subseteq D(W)$.
When the class is indexable, the Whittle's index in state $n$ in class-$k$ is defined as: \[W_k(n)=\min \{W |n \in D(W)\}\]
\end{mydef}
We start showing the indexability of the problem. The proof is straightforward and can be obtained from the previous work in this area. 

\begin{proposition}\label{prop_wi}
Assuming that for each $W$ and class-$k$, the optimal solution is of type threshold $n_k(W)$, and $\sum_0^n u^n_k(q)$ is increasing in $n$, then the class is indexable. 
\end{proposition}
Where $n_k(W)$ is an optimal threshold at W(i.e. optimal solution of the relaxed problem for given W) in class-$k$ and $u^n_k$ the stationary distribution of the queue states under threshold policy $n$ in class-$k$. One can show that the condition in the aforementioned proposition is satisfied for our problem. The proof is similar to the one in \cite{kriouile} and is skipped here for brevity.\\ 
Several works have been conducted in the past to find Whittle index values for different scheduling problems, e.g. \cite{larranaga15} and the references therein. In \cite{larranaga15}, an algorithm is provided to compute Whittle's index for a queueing system with one server. This algorithm in fact, gives recursively expression of whittle index for given state. However, the complexity of the algorithm grows with number of states and hence, cannot be practically applied to our context. More generally, this algorithm cannot be theoretically applied in some cases where the passive decision's average time takes different values for an infinite set of states, since the number of iterations of the algorithm will be infinite. Since we consider that the queue state is not bounded, the above algorithm cannot be used. In \cite{kriouile}, a closed form expression of Whittle's index is given, which simplifies the complexity of the computation.  
Let us now restate the Whittle's index result in \cite{kriouile}. 

\begin{proposition}\label{prop:w.i} \cite{kriouile}
The Whittle's indices expressions are defined for states $\in$ $[0,R_k-1]$ and are given as,\\
for $0 \leq n \leq R_k-1$: $W_k(n)=w_n^k=x_{n,n-1}^k=\frac{a_kR_kn}{R_k-n}$
\end{proposition}

Based on the above Whittle's index result, the work in \cite{kriouile} provides a Whittle index policy, which consists on allocating the servers to the $M$ users which have the highest Whittle index at time $t$, denoted by $WI$.\\
However, the above result is limited to the case where the states are $[0,R_k-1]$. The technique used in \cite{kriouile} consists of finding the stationary distribution of the states under threshold policy, reformulating the relaxed problem using this stationary distribution and analyzing this reformulated problem (which is similar to a deterministic one) to find the explicit expressions of Whittle index based on the algorithm that we have discussed before. In this paper, since the algorithm that gives us the whittle index expressions can not be applied for all states, we rely on another method which allows us to find the Whittle index values for all possible states. In order to work with the original cost function, we formulate a discounted reward problem in which $\beta$ is a discount factor. We analyze this discounted problem and found the Whittle index expressions (that depend on $\beta$) and then by taking $\beta \rightarrow 1$, we obtain the Whittle's index for our original problem.   

\subsection{Problem reformulation using discounted cost approach} \label{sec:disc_cost}
We start by formulating the original problem with the expected discounted cost:
\begin{equation}\label{eq:constraint_discounted}
\begin{aligned}
& \underset{\phi \in \Phi}{\text{min}}&& E \left[ \sum_{t=0}^{+\infty}  \sum_{k=1}^{K} \sum_{i=1}^{{\gamma_k}N} \beta^t a_k q_i^k(t) \mid q(0), \phi\right], \\
& \text{s.t.} && \sum_{k=1}^{K} \sum_{i=1}^{{\gamma_k}N} s_i^k(t)\leq {\alpha}N, \forall t.\\
\end{aligned}
\end{equation}
Following the same steps as in section \ref{sec:probForm}, we relax the problem and give the dual relaxed problem for given $W$:
\begin{equation}\label{eq:relaxed_dual_discounted_cost_individual}
\underset{\phi \in \Phi}{\text{min}} \sum_{t=0}^{+\infty} E[ \sum_{k=1}^{K} \sum_{i=1}^{{\gamma_k}N}\beta^t(a_k q_i^k(t)+W s_i^k(t)) \mid \phi,q(0)].
\end{equation}

Then we decompose it into $N$ individual problems since  the Bellman equation that resolves the dual problem is decomposable.  
The Bellman equation for an individual problem is \cite{ross}.  
\begin{equation} \label{eq:individual_discounted_cost}
V(q_i^k)= \underset{s_i^k}{\text{min}} \{ C(q_i^k,s_i^k)+ \beta \sum_{q^{'k}_i} Pr(q^{'k}_i|q_i^k,s_i^k)V(q^{'k}_i)\}
\end{equation}
In fact $V(q_i^k)$ is no more than the discounted cost when the initial state is $q_i^k$, $V(q_i^k)=\sum_{t=0}^{+\infty} E[ \beta^t(a_kq_i^k(t)+Ws_i^k(t)) \mid \phi,q_i^k(0)=q_i^k]$, and $C(q_i^k,s_i^k)=a_kq_i^k+Ws_i^k$.

Following the same method in section \ref{sec:threshold policy}, we can prove that the optimal solution that satisfies this Bellman equation is a threshold policy, by proving that the function $TO(V)$ is submodular. We can also conclude that the value function has same structural property as in section \ref{sec:threshold policy}, especially that the submodularity and $R_k$-convexity hold true.
However, contrary to what has been done in \cite{kriouile}, finding the steady state distribution will not give an explicit expression of the problem \ref{eq:relaxed_dual_discounted_cost_individual}. Nevertheless, we can work only with the Bellman equation to derive the Whittle index thanks to the parameter $\beta$ which helps us to find the Whittle index for all states. 

\begin{mydef} 
We define $C_0^n(q_i^k)$ and $C_1^n(q_i^k)$ in class-$k$ as the discounted costs starting at the initial queue state $q_i^k$ at which the decision taken is to not be scheduled ($s_i^k=0$) or to be scheduled ($s_i^k=1$) respectively and when the policy considered is threshold $n$, explicitly: 
$$C_0^n(q_i^k)\triangleq a_kq_i^k+ \beta \sum_{q^{'k}_i} Pr(q^{'k}_i|q_i^k,0)V^n(q^{'k}_i)$$
$$C_1^n(q_i^k)\triangleq a_kq_i^k+W + \beta \sum_{q^{'k}_i} Pr(q^{'k}_i|q_i^k,1)V^n(q^{'k}_i)$$
Where $V^n(\cdot)$ is the value function under threshold policy $n$.
\end{mydef}

\begin{mydef}\label{function_g} 
We define $g_k(n,W)$ as function defined in $[0,+\infty[ \times \mathbb{R}$, such that for all $(n,W) \in [0,+\infty[ \times \mathbb{R}$, $g_k(n,W)=C_1^n(n)-C_0^n(n)$ 
\end{mydef}

\begin{proposition}\label{prop:expression_of_g(n,w)}
For $0 \leq n \leq R_k-1$: $$g_k(n,W)=W(1-n \beta \rho_k)-a_kn \beta$$
For $n \geq R_k$: $$g_k(n,W)=\frac{W(1-\beta)-a_kR_k \beta}{1- \rho_k \beta}$$
\end{proposition}

\begin{proof}
See Appendix \ref{app:expression_of_g(n,w)}
\end{proof}

We emphasize that to prove that for fixed W, in class-$k$, a given state $n$ is indeed an optimal threshold (i.e. if $q_i^k \leq n$ the queue is not scheduled and otherwise it is scheduled), we  just need to prove that it satisfies for all states $q_i^k \leq n$ $C_0^n(q_i^k) \leq C_1^n(q_i^k)$ and for $q_i^k > n$ $C_0^n(q_i^k) \geq C_1^n(q_i^k)$ (according to Bellman equation). In other words, we suppose that $n$ is a threshold (i.e. if $q_i^k \leq n$ the queue is not scheduled and otherwise it is scheduled), and we show that for all states $q_i^k \leq n$ $C_0^n(q_i^k) \leq C_1^n(q_i^k)$ and for $q_i^k > n$ $C_0^n(q_i^k) \geq C_1^n(q_i^k)$ (for given value of $W$, the optimal threshold might not be unique).

\begin{proposition}\label{prop:optimality_condition_of_n}
For class-$k$, if there exists $n$ such that $C_0^n(n)=C_1^n(n)$, then $n$ is an optimal threshold.
\end{proposition}

\begin{proof}
See Appendix \ref{app:optimality_condition_of_n}
\end{proof}

\begin{proposition}\label{prop:threshold_condition}
If $W=\frac{\beta a_k R_k n}{R_k- \beta n}$, then for $n \leq R_k-1$, $n$ is an optimal threshold. And if $W = \frac{a_kR_k \beta}{1-\beta}$, then for all $n \geq R_k$ $n$ is an optimal threshold. 
\end{proposition}

\begin{proof}
See Appendix \ref{app:threshold_condition}   
\end{proof}

In order to establish the Whittle indices we study the function $g_k$ defined in definition \ref{function_g}.
\begin{Lemma}
$g_k$ is strictly increasing in $W$, and decreasing in $n$.
\end{Lemma}
\begin{proof}
$g_k$ is clearly strictly increasing in $W$ and decreasing in $n$ from its expression.
\end{proof}    

To prove that the Whittle index for a given state $n$ is a given $W_k(n)$ in class-$k$, we have to demonstrate that for all $W \leq W_k(n)$, at state $n$ the decision must be the active action. In other words, since the optimal solution is surely a threshold policy, we need to prove that for all states greater than $n$, they cannot be the optimal threshold. For that, we will suppose that if the optimal threshold is higher than $n$, including the case of infinite threshold, and we will prove that there is a contradiction.

\begin{proposition}\label{prop:finite_threshold_condition}
If $W < \frac{a_kR_k \beta}{1-\beta}$, then the optimal threshold is surely finite.
\end{proposition}
\begin{proof}
See Appendix \ref{app:finite_threshold_condition}
\end{proof}

\begin{proposition}\label{prop:w.i}
For each queue state $n$ in class-$k$, the Whittle index expression is given by:\\
For $0 \leq n \leq R_k-1$, $W_k(n)=\frac{\beta a_k R_k n}{R_k- \beta n}$.\\ 
For $n \geq R_k$, $W_k(n)=\frac{a_kR_k \beta}{1-\beta}$. 
\end{proposition}
\begin{proof}
See Appendix \ref{app:w.i}
\end{proof}

We know that for $\beta \rightarrow 1$, the solution for the problem \ref{eq:relaxed_dual_discounted_cost_individual} is the same as the problem \ref{eq:relaxed}, see \cite{ross}.
Hence, to derive the Whittle index for the expected average cost's case, we must tend $\beta$ to $1$. However, for states greater or equal than $R$, the Whittle indices tend to $+\infty$. 
On the other hand, by looking at our policy which consists on selecting the users at states with the $M$ highest Whittle index values, we can notice that this policy is the same if the order of the Whittle indices from the biggest to the smallest one is not affected even if the Whittle index values are modified. In the following, we denote by $W_k(n)$ the Whittle index of state $n$ at class-k.
\begin{thm}\label{prop:new_policy} 
For any $\beta > 1-\frac{\min\{a_j R_j\}}{\max\{a_j R_j^2\}}$, the Whittle index policy where the Whittle indices in each class $k$ are given by:\\
For $0 \leq n < R_k$: $W_k(n)=\frac{\beta a_k R_k n}{R_k- \beta n}$\\  
For $R_k \leq n$: $W_k(n)=a_k R_k \max \{a_j R_j^2\}$\\
 is exactly the Whittle index policy when the Whittle indices are given by proposition \ref{prop:w.i}.
\end{thm}
\begin{proof}
See Appendix \ref{app:new_policy}
\end{proof}         
When $\beta \rightarrow 1$, the condition given in Theorem \ref{prop:new_policy} is still true, and hence, we get the Whittle index policy for our original problem with expected average cost. The policy consists in allocating the channels (or servers) to the $M$ users having the highest $M$ Whittle index values computed in the aforementioned theorem. 
We can notice that the Whittle index of states greater than the maximum transmission rate are different only by $a_k$ and $R_k$ ($W_k(n)=a_k R_k \max \{a_j R_j^2\}$). It is worth mentioning that the obtained policy can be seen as $c \mu$ rule when all states are greater than $R_k$, since we choose $M$ users with the highest $a_k R_k$. 
\section{Numerical Results}\label{sec:numerics}
In this section, we show that the Whittle index policy (denoted by WI) shows good performance when the number of users is large. Since computing the optimal policy is computationally prohibitive, we take advantage of the fact that the the optimal cost of the relaxed problem denoted by $C^{RP,N}$ is less than the optimal one of the original problem. We therefore, compare between the cost obtained by our policy and the one obtained for the relaxed problem, for which a simple threshold policy is the optimal one. The gap between these two policies is then an upper bound of the gap between our policy and the optimal one (in terms of achieved average cost). In addition, we compare the average cost given by our policy WI with the one given by the myopic policy or the Max-Weight which schedule the $M$ queues that have the highest instantaneous incurred delay cost. We denote $C^{WI,N}$ the average cost given by the policy WI and $C^{MD,N}$ the average cost given by the myopic policy. We plot the results on Figures 1 and 2 where we consider two user classes of users with their respective transmission rate $R_1$ and $R_2$, and the number of servers is equal to $N/2$ where N is the number of users. In Figure 1 we take $R_1=5$ and $R_2=20$,  while  in Figure 2 the values are $R_1=10$ and $R_2=45$. According to these figures, one can show that our policy WI is asymptotically optimal, and performs much better than the myopic policy. This confirms our main motivation behind developing the Whittle index policy as presented before in the paper.

\begin{figure}
\centering
\includegraphics[width=0.995\linewidth]{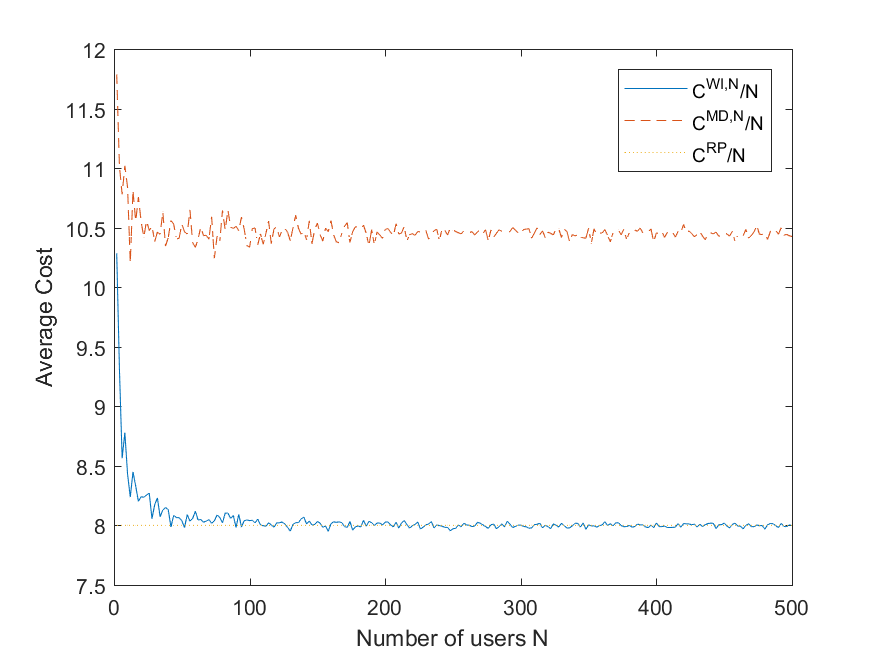}
\caption{Comparison between the average costs of the proposed policy WI, the optimal policy of relaxed problem, and the myopic policy}
\end{figure}

\begin{figure}
\centering
\includegraphics[width=0.995\linewidth]{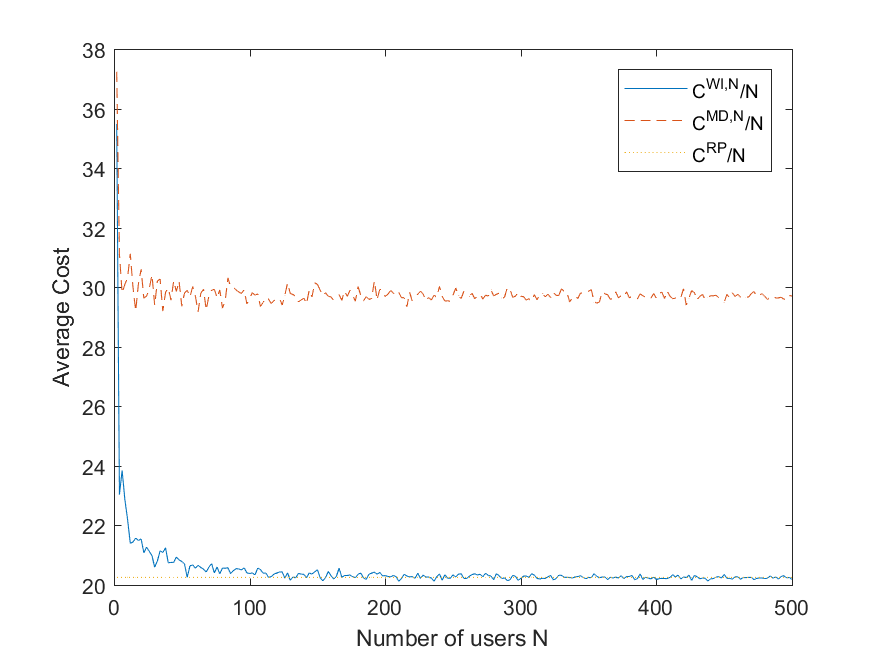}
\caption{Comparison between the average costs of the proposed policy WI, the optimal policy of relaxed problem, and the myopic policy}
\end{figure}

\section{Conclusion}
We considered the problem of resource allocation in a queueing system composed of $N$ queues and $M$ servers. We showed that minimizing the average expected queue length, in a time discrete slotted system, is a Restless Bandit Problem, for which finding the optimal solution is out of reach. We therefore developed a simple Whittle index policy for this problem. While the previous works on Whittle index for time discrete queueing systems have been mainly limited to the context of finite buffer length models, we provided in this paper an extension and derived an index policy without restricting the buffer length to be always less than an fixed value. Our development rely on the idea to introduce a new discount factor in the cost function, to derive the Whittle index as function of this discount factor and then obtain the index value of the original problem by taking the limit when this factor tends to 1. Numerical results show that our policy is asymptotically optimal in the many user regime. 


\begin{appendices}

\section{Proof of proposition \ref{prop:expression_of_g(n,w)}}\label{app:expression_of_g(n,w)}
1) $n \leq R_k-1$:\\
We start first by giving a useful lemma
\begin{Lemma}\label{lem:velue_function_difference}
For all $0 \leq i \leq p \leq R_k-1$, $$V^p(R_k+i)-V^p(i)=a_kR_k+W$$
\end{Lemma}
\begin{proof}
We decompose the discounted cost $V^p(i+R_k)$ in the cost incurred at first time slot plus the discounted cost starting at the next time slot. At state $i+R_k$, the decision taken is to transmit since $i+R_k \geq R_k > p$, and at state $i$, the decision taken is passive action since $i \leq p$, hence,
$$V^p(i+R_k)=a_k(i+R_k)+W+ \rho_k \beta \sum_{j=0}^{R_k-1} V^p(i+j)$$
$$V^p(i)=a_ki+ \rho_k \beta \sum_{j=0}^{R_k-1} V^p(i+j)$$
Subtracting the second term from the first term,
$$V^p(i+R_k)-V^p(i)=a_kR_k+W$$
\end{proof}
We have
$$C_1^n(n)=a_kn+W+ \beta \rho_k \sum_{i=0}^{R_k-1} V^n(i)$$
$$C_0^n(n)=a_kn + \beta \rho_k \sum_{i=0}^{R_k-1} V^n(i+n)$$
That means,
$$C_1^n(n)-C_0^n(n)=W+ \beta \rho_k \sum_{i=0}^{R_k-1} V^n(i) -\beta \rho_k \sum_{i=0}^{R_k-1}V^n(i+n)$$
$$=W+ \beta \rho_k \sum_{i=0}^{R_k-1} V^n(i)- \beta \rho \sum_{i=n}^{R_k-1+n} V^n(i)$$
$$=W+ \beta \rho_k \sum_{i=0}^{n-1} V^n(i)- \beta \rho_k \sum_{i=R_k}^{R_k-1+n} V^n(i)$$
$$=W+ \beta \rho_k \sum_{i=0}^{n-1} V^n(i)- \beta \rho_k \sum_{i=0}^{n-1} V^n(i+R_k)$$
Applying the lemma \ref{lem:velue_function_difference},
$$C_1^n(n)-C_0^n(n)=W - \beta \rho_k n (a_kR_k+W)$$
$$C_1^n(n)-C_0^n(n)=W(1-\beta \rho_k n) - \beta n a_k $$

2)$n\geq R_k$:\\
We consider a given threshold $n \geq R_k$, i.e., for states less than $n$, we don't transmit otherwise we transmit.
At state $n$ if we decide to transmit then the next possible states are $n-R_k+i$ ($i$ varies from $0$ to $R_k-1$) with the probability to reach each state is $\rho_k$, hence, we have,
\[C_1^n(n)=a_kn+ W+ \beta \rho_k \sum_{i=0}^{R_k-1} V^n(n-R_k+i)\]
At state $n+i-R_k$, since $n-R_k+i < n$ then, the decision taken is passive action ($n$ is threshold), thus if we decompose again  $V^n(n-R_k+i)$, $V^n(n+i-R_k)= a_k (n+i-R_k)+\beta \rho_k \sum_{j=0}^{R_k-1} V^n(n-R_k+i+j)$, Replacing $V^n(n+i-R_k)$ by its value,
\[C_1^n(n)=a_kn+W +\beta \rho_k \sum_{i=0}^{R_k-1} [a_k(n+i-R_k)+ \rho_k \beta \sum_{j=0}^{R_k-1} V^n(n-R_k+i+j)]\]
\[C_1^n(n)=a_kn+W +\beta \rho_k \sum_{i=1}^{R_k-1} [a_k(n+i-R_k)+ \rho_k \beta \sum_{j=0}^{R_k-1} V^n(n-R_k+i+j)]\] 
\[+\beta \rho_k[a_k(n-R_k)+ \rho_k \beta \sum_{j=0}^{R_k-1} V^n(n-R_k+j)]\]
We know that,
$$C_1^n(n)=a_kn+ W+ \beta \rho_k \sum_{j=0}^{R_k-1} V^n(n-R_k+j)$$
Hence,
\[C_1^n(n)=a_kn+W+ \beta \rho_k \sum_{i=1}^{R_k-1} [a_k(n+i-R_k)+ \rho_k \beta \sum_{j=0}^{R_k-1} V^n(n-R_k+i+j)]\]
\[+\beta \rho_k C_1^n(n)- \rho_k \beta W - a_k \beta\]
That means,
$$C_1^n(n)=\frac{1}{1- \rho_k \beta}[a_kn+W+ \beta \rho_k \sum_{i=1}^{R_k-1} [a_k(n+i-R_k)$$ $$+ \rho_k \beta \sum_{j=0}^{R_k-1} V^n(n-R_k+i+j)]- \rho_k \beta W - a_k \beta]$$
$$C_1^n(n)=\frac{1}{1- \rho_k \beta}[a_kn+W +\beta \rho_k \sum_{i=1}^{R_k-1} a_k(n+i-R_k)$$ $$+ \rho_k^2 \beta^2 \sum_{i=1}^{R_k-1}\sum_{j=0}^{R_k-1} V^n(n-R_k+i+j)- \rho_k \beta W - a_k \beta]$$

At state $n$ if we decide to not transmit then the next possible states are $n+i$ with the probability to reach each state is $\rho_k$ hence, we have,
$C_0^n(n)=an+ \beta \rho_k \sum_{i=0}^{R_k-1} V^n(n+i)$
At state $n+i$ for $i > 0$, since $n+i > n$ then, the decision taken is active action ($n$ is threshold), thus if we decompose again $V^n(n+i)$, $V^n(n+i)= a_k (n+i)+W+\beta \rho_k \sum_{j=0}^{R_k-1} V^n(n-R_k+i+j)$, Replacing $V^n(n+i)$ when $i>0$ by its value,
\[C_0^n(n)=a_kn+ \beta \rho_k \sum_{i=1}^{R_k-1} [a_k(n+i)+W\]
\[+\rho_k \beta \sum_{j=0}^{R_k-1} V^n(n-R_k+i+j)]+\beta \rho_k V^n(n)\] 
However $V^n(n)$ is no more than $C_0^n(n)$, hence,
\[C_0^n(n)=a_kn+ \beta \rho_k \sum_{i=1}^{R_k-1} [a_k(n+i)+W\]
\[+ \rho_k \beta \sum_{j=0}^{R_k-1} V^n(n-R_k+i+j)]+\beta \rho_k C_0^n(n)\] 
That means,
$$C_0^n(n)=\frac{1}{1-\rho_k \beta}[a_kn+ \beta \rho_k \sum_{i=1}^{R_k-1} [a_k(n+i)+W$$ $$+ \rho_k \beta \sum_{j=0}^{R_k-1} V^n(n-R_k+i+j)]]$$
$$C_0^n(n)=\frac{1}{1-\rho_k \beta}[an+ \beta \rho_k \sum_{i=1}^{R_k-1} [a_k(n+i)+W]$$ $$+ \rho_k^2 \beta^2 \sum_{i=1}^{R_k-1}\sum_{j=0}^{R_k-1} V^n(n-R_k+i+j)]$$

Then,
\begin{align*}
(C_1^n(n)-C_0^n(n)) (1-\rho_k \beta)=&  W - \beta \rho_k \sum_{i=1}^{R_k-1} [a_kR+W]\\
&- \rho_k \beta W - a_k \beta\\
(C_1^n(n)-C_0^n(n)) (1-\rho_k \beta)=&  W - \beta (R_k-1) a_k\\ 
&- \beta \rho_k (R_k-1)W - \rho_k \beta W - a_k \beta\\
(C_1^n(n)-C_0^n(n)) (1-\rho_k \beta)=&  W - \beta a_kR_k- \beta W\\
(C_1^n(n)-C_0^n(n)) (1-\rho_k \beta)=& W(1- \beta) - \beta a_kR_k 
\end{align*}
Hence,
$C_1^n(n)-C_0^n(n)= \frac{W(1- \beta) - \beta a_kR_k}{1-\rho_k \beta}$
\section{Proof of proposition \ref{prop:optimality_condition_of_n}}\label{app:optimality_condition_of_n}

As the function $TO(V^n)$ is submodular, then for all $q_i^k \leq n$, we have $C_0^n(q_i^k)-C_1^n(q_i^k) \leq  C_0^n(n)-C_1^n(n)=0$, and  for all $q_i^k>n$, we have $C_0^n(q_i^k)-C_1^n(q_i^k) \geq  C_0^n(n)-C_1^n(n)=0$. 
That means $n$ is indeed an optimal threshold. 
\section{Proof of proposition \ref{prop:threshold_condition}}\label{app:threshold_condition}
1)$n \leq R_k-1$:\\
According to proposition \ref{prop:expression_of_g(n,w)} for $n \leq R_k-1$, $C_1^n(n)-C_0^n(n)=g_k(n,W)=W(1-n \beta \rho_k)-a_k \beta n$.
Knowing that $g_k(n,W)=0 \Leftrightarrow W=\frac{\beta a_k R_k n}{R_k- \beta n}$.    
That means, for $W=\frac{\beta a_k R_k n}{R_k- \beta n}$, $C_1^n(n)-C_0^n(n)=0$
Hence, using proposition \ref{prop:optimality_condition_of_n}, $n$ is indeed an optimal threshold.\\
2)$n \geq R_k$:\\
According to proposition \ref{prop:expression_of_g(n,w)} for $n \geq R_k$, $C_1^n(n)-C_0^n(n)=g_k(n,W)=\frac{W(1-\beta)-a_k \beta R_k}{1-\rho_k \beta}$.
Knowing that $g_k(n,W)=0 \Leftrightarrow W=\frac{\beta a_k R_k}{1- \beta}$.
Hence, 
$C_1^n(n)=C_0^n(n) \Leftrightarrow W= \frac{a_kR_k \beta}{1-\beta}$.
Applying proposition \ref{prop:optimality_condition_of_n}, the threshold $n$ is indeed an optimal solution when $W = \frac{a_kR_k \beta}{1 - \beta}$. That is true for all $n \geq R_k$, which concludes the proof.
\section{Proof of proposition \ref{prop:finite_threshold_condition}}\label{app:finite_threshold_condition}
We consider that the optimal solution is an infinite threshold and we consider  $V^{\infty}$ is the value function under infinite threshold.
\begin{Lemma}\label{lem:ineq_infinite_val_funct}
For all $q_i^k$ $V^{\infty}(q_i^k+R_k)-V^{\infty}(q_i^k) \geq \frac{a_kR_k}{1-\beta}$.
\end{Lemma}
\begin{proof}
Under infinite threshold, since the decision taken for all states is passive action, then we have for all $q_i^k \geq 0$,
\[V^{\infty}(q_i^k+R_k)= a_k(q_i^k+R_k) -W + \beta \rho_k \sum_{i=1}^{R_k-1} V^{\infty}(q_i^k+R_k+i)\]
and
$V^{\infty}(q_i^k)= a_kq_i^k -W + \beta \rho_k \sum_{i=1}^{R_k-1} V^{\infty}(q_i^k+i)$
Hence,
\[V^{\infty}(q_i^k+R_k)-V^{\infty}(q_i^k)= a_kR_k + \beta \rho_k \sum_{i=1}^{R_k-1} [V^{\infty}(q_i^k+R_k+i)-V^{\infty}(q_i^k+i)]\]
Due to the $R_k$-convexity of $V^{\infty}$, then $V^{\infty}(q_i^k+R_k+i)-V^{\infty}(q_i^k+i)>V^{\infty}(q_i^k+R_k)-V^{\infty}(q_i^k)$
hence, $V^{\infty}(q_i^k+R_k)-V^{\infty}(q_i^k) \geq a_kR_k + \beta  [V^{\infty}(q_i^k+R_k)-V^{\infty}(q_i^k)]$,
that implies  $V^{\infty}(q_i^k+R_k)-V^{\infty}(q_i^k) \geq \frac{a_kR_k}{1- \beta}$.
\end{proof}    
We have the difference between $C_1^{\infty}(q_i^k)$ and $C_0^{\infty}(q_i^k)$ for $q_i^k \geq R_k$,
\[C_1^{\infty}(q_i^k)-C_0^{\infty}(q_i^k)=W+ \beta \rho_k \sum_{i=0}^{R_k-1} [V^{\infty}(q_i^k+i-R_k)-V^{\infty}(q_i^k+i)]\]
Applying lemma \ref{lem:ineq_infinite_val_funct}
\[C_1^{\infty}(q_i^k)-C_0^{\infty}(q_i^k) \leq W+ \beta \rho_k \sum_{i=0}^{R_k-1} [-\frac{a_kR_k}{1- \beta}]\]
According to proposition 8's assumption, we have that $W < \frac{a_kR_k \beta}{1- \beta}$, therefore, 
\[C_1^{\infty}(q^k_i)-C_0^{\infty}(q^k_i) < \frac{a_kR_k \beta}{1- \beta} - \frac{a_kR_k \beta}{1- \beta} = 0\]
Thus, at state $q^k_i$, the optimal decision is to transmit which contradict with the fact that the threshold is infinite.
\section{Proof of proposition \ref{prop:w.i}}\label{app:w.i}
1) $0 \leq n \leq R_k-1$:\\
We fix a state $n$ less than $R_k-1$.
For  $W < \frac{a_kR_k n \beta}{R_k-\beta n}$, as $g_k$ is strictly increasing in $W$, then $g_k(n,W)$ is less strictly than $0$, as $g_k$ is decreasing in $n$, then for all $q_i^k \geq n$, $g(q_i^k,W) \leq g_k(n,W) < 0$,  that implies for all $q_i^k \geq n$, $q_i^k$ can not be threshold( otherwise there exists $q_i^k \geq n$ such that $g_k(q_i^k,W) \geq 0$). Moreover, since $W < \frac{a_kR_k n \beta}{R_k-\beta n} \leq \frac{aR \beta}{1-\beta}$, then according to proposition \ref{prop:finite_threshold_condition}, optimal threshold must be finite. Hence, surely the optimal threshold is less strictly than $n$. That means at state $n$ the optimal decision is active action, hence, $n \notin D(W)$.
Applying the proposition \ref{prop:threshold_condition}, for $W = \frac{a_kR_k n \beta}{R_k-\beta n}$, the threshold $n$ is an optimal solution. Then $n \in D(W)$ Hence, we conclude the result for the first case.\\
2) $n \geq R_k$:
The second case comes from propositions \ref{prop:threshold_condition} and \ref{prop:finite_threshold_condition}.
In fact we need to prove that for all $n \geq R_k$, $\frac{a_kR_k \beta}{1-\beta}= \min \{W, n \in D(W)\}$, in other word for all $W < \frac{a_kR_k \beta}{1-\beta}$, $n \notin D(W)$ ( the optimal decision is active), and that $n \in D(\frac{a_kR_k \beta}{1-\beta})$.
for  $W < \frac{a_kR_k \beta}{1-\beta}$, as $g_k$ is strictly increasing in $W$, then $g_k(n,W)$ is less strictly than $0$ for all $n \geq R_k$, that implies for all $n \geq R_k$, $n$ can not be threshold. Moreover according to proposition \ref{prop:finite_threshold_condition}, threshold must be finite, hence, surely the optimal threshold is less strictly than $R_k$, that means at state $n \geq R_k$ the optimal decision is active action, hence, $n \notin D(W)$.
And applying the proposition \ref{prop:threshold_condition}, for  $W = \frac{a_kR_k \beta}{1-\beta}$, the threshold $n$ is an optimal solution, then $n \in D(W)$. Hence, we conclude the result.

\section{Proof of Theorem \ref{prop:new_policy}}\label{app:new_policy}
In order to prove this result, we need to prove that the order from the biggest Whittle index to the smallest one is the same.\\
First of all, we give the order of the Whittle indices given by propositions \ref{prop:w.i}, when $\beta > 1-\frac{\min\{a_j R_j\}}{\max\{a_j R_j^2\}}$.\\
Due to the indexability of all classes, it's obvious that the Whittle index is increasing in $n$ for given class $k$.
Moreover, considering any two classes $k$ and $m$, then, the Whittle index of any state $n_k$ in class $k$ and  greater than $R_k$ is larger than the Whittle index of any state $n_m$ in class $m$ less than $R_m$.
In fact we have $\beta > 1-\frac{\min\{a_j R_j\}}{\max\{a_j R_j^2\}}$, that means $1-\beta < \frac{\min\{a_j R_j\}}{\max\{a_j R_j^2\}}$, hence, $\frac{1}{1-\beta} > \frac{\max\{a_j R_j^2\}}{\min\{a_j R_j\}}$, hence,  $W_k(n_k)=\frac{a_k R_k}{1-\beta} > \frac{a_k R_k \max\{a_j R_j^2\}}{\min\{a_j R_j\}}\geq \max\{a_j R_j^2\} \geq a_m R_m^2 \geq a_m R_m (R_m-1)\geq \beta a_m R_m (R_m-1) \geq \frac{\beta a_m R_m (R_m-1)}{R_m- \beta (R_m-1)}$(because $R_m- \beta (R_m-1) \geq 1$) $=W_m(R_m-1) \geq W_m(n_m)$.\\
For the new form of Whittle index where we get rid of $\beta$ for states greater than maximum transmission rate, the order is not affected for the states less than $R_k-1$, since the Whittle indices are the same.
For the states greater than $R_k$, the Whittle index  $a_k R_k \max \{a_j R_j^2\}$ is higher than $a_k R_k a_m R_m (R_m-1) \geq a_m R_m (R_m-1) \geq \frac{\beta a_m R_m (R_m-1)}{R_m- \beta (R_m-1)}=W_m(R_m-1)$
Furthermore, the order between the Whittle indices greater than the transmission rate doesn't change since we just multiply by a constant which is $\frac{1-\beta}{\max\{a_j R_j^2\}}$ to go from $\frac{a_k R_k}{1-\beta}$ to  $a_k R_k \max \{a_j R_j^2\}$.
Hence, the new form of Whittle doesn't affect the order of Whittle index. That means, WI gives us the same Whittle index policy.

\end{appendices}
\end{document}